\documentclass[a4paper,onecolumn,unpublished]{quantumarticle}
\pdfoutput=1

\usepackage{amsmath,amssymb,amsthm,mathtools,bm,bbm}

\usepackage{graphicx}
\usepackage{subcaption}
\usepackage{booktabs}
\usepackage{float}
\usepackage{braket} 

\usepackage{xcolor}

\usepackage[T1]{fontenc}
\usepackage{microtype}

\usepackage[numbers,sort&compress]{natbib}
\usepackage[hidelinks,unicode]{hyperref} 
\usepackage[nameinlink]{cleveref}
\usepackage{doi} 

\usepackage{algorithm}
\usepackage{algpseudocode}
\usepackage{enumitem}
\setlist[itemize]{leftmargin=*,topsep=2pt}
\setlist[enumerate]{leftmargin=*,topsep=2pt}

\numberwithin{equation}{section}
\newtheorem{theorem}{Theorem}[section]
\newtheorem{lemma}[theorem]{Lemma}
\newtheorem{proposition}[theorem]{Proposition}
\newtheorem{corollary}[theorem]{Corollary}
\newtheorem{assumption}[theorem]{Assumption}

\title{PAC global optimization for VQE in low-curvature geometric regimes}
\author{Benjamin Asch}
\affiliation{ETH Z\"urich}
\email{benjamintasch@berkeley.edu}

\newcommand{\R}{\mathbb{R}}
\newcommand{\T}{\mathbb{T}}
\newcommand{\Tp}{\T^{p}}
\newcommand{\B}{\mathbb{B}}

\newcommand{\Prb}{\mathbb{P}}

\DeclareMathOperator{\poly}{poly}

\newcommand{\Cost}{C}

\newcommand{\Z}{\mathbb{Z}}

\newcommand{\rad}{\mathrm{rad}}

\newcommand{\Ball}[2]{\B\!\left(#1,#2\right)}

\newcommand{\wh}[1]{\widehat{#1}}

\newcommand{\dist}{\mathrm{dist}}

\DeclareMathOperator{\Gr}{\mathrm{Gr}}

\newcommand{\EE}{\mathbb{E}}

\newcommand{\HH}{\mathcal{H}}     

\newcommand{\Unif}{\mathrm{Unif}}

\theoremstyle{definition}
\newtheorem{definition}{Definition}[section]

\theoremstyle{remark}
\newtheorem{remark}{Remark}[section]

\theoremstyle{definition}
\newtheorem{example}{Example}[section]

\crefname{theorem}{Theorem}{Theorems}
\crefname{lemma}{Lemma}{Lemmas}
\crefname{algorithm}{Algorithm}{Algorithms}
\crefname{assumption}{Assumption}{Assumptions}
\crefname{definition}{Definition}{Definitions}
\crefname{section}{Section}{Sections}
\crefname{remark}{Remark}{Remarks}

\begin{document}
\maketitle

\begin{abstract}
We give noise-robust, Probably Approximately Correct (PAC) guarantees of
global $\varepsilon$-optimality for the Variational Quantum Eigensolver under explicit geometric conditions. For periodic ansatzes
with bounded generators—yielding a globally Lipschitz cost landscape on a toroidal
parameter space—we assume that the low-energy region containing the global
minimum is a Morse--Bott submanifold whose normal Hessian has rank
$r = O(\log p)$ for $p$ parameters, and which satisfies polynomial fiber regularity with respect to coordinate-aligned, embedded flats. This low-curvature-dimensional structure serves as a model for regimes in which only a small number of directions control energy variation, and is consistent with mechanisms such as strong parameter tying together with locality in specific multiscale and tied shallow architectures.

Under this assumption, the sample complexity required to find an
$\varepsilon$-optimal region with confidence $1-\delta$ scales with the curvature dimension $r$
rather than the ambient dimension $p$.  With probability at least $1-\delta$, the
algorithm outputs a region in which all points are $\varepsilon$-optimal, and at least one lies within a bounded neighborhood of the global minimum. The resulting complexity is quasi-polynomial in $p$ and
$\varepsilon^{-1}$ and logarithmic in $\delta^{-1}$. This identifies a
geometric regime in which high-probability global optimization remains
feasible despite shot noise.
\end{abstract}


\section{Introduction}
\label{sec:intro}

Variational quantum algorithms (VQAs)~\cite{Cerezo_2021_VQAs} remain among the promising near-term approaches for leveraging noisy intermediate-scale
quantum (NISQ) devices~\cite{Bharti_2022, Preskill_2018, farhi2014quantumapproximateoptimizationalgorithm}. They operate by
parameterizing a family of quantum circuits and optimizing an expected energy or
cost functional. Among the most prominent examples, the Variational Quantum
Eigensolver (VQE)~\cite{Peruzzo_2014, Tilly_2022} has achieved empirical success
across quantum chemistry and condensed-matter
applications~\cite{Kandala_2017, Hempel_2018}.
Despite this progress, the theoretical understanding of variational optimization
remains incomplete. The induced landscapes are generally nonconvex and
high-dimensional, often shaped by symmetries, degeneracies, and noise that
complicate the geometry of global minima~\cite{anschuetz2023criticalpointsquantumgenerative, Anschuetz_2022}.

Several analyses capture restricted regimes. When the ansatz map is locally surjective on the relevant manifold, gradient flow can converge to the global minimum given infinite precision and time~\cite{wiedmann2025convergencevariationalquantumeigensolver}, though this condition is rarely met by hardware-efficient circuits. Conversely, in sufficiently expressive families, gradients can vanish exponentially with system size, producing the barren-plateau phenomenon that
renders local descent ineffective~\cite{McClean_2018, Cerezo_2021_BPs}. 
More recently, it has been shown that shallow alternating-layer circuits with
guiding states and small initialization admit a linearized training
interpretation, enabling convergence and generalization guarantees across input
families~\cite{nguyen2025theoreticalguaranteesvariationalquantum}, albeit under a
warm-start requirement. 
Global or gradient-free optimizers---Bayesian, evolutionary, or
heuristic~\cite{ShahriariSWAF16, hansen2023cmaevolutionstrategytutorial}---may
perform well empirically, but typically lack rigorous non-asymptotic guarantees
and scale poorly with the number of parameters. 
Classical Lipschitz-based global optimization
methods~\cite{malherbe2017globaloptimizationlipschitzfunctions} produce correct solutions but rely on uniform partitions of
the ambient space, yielding exponential dependence on the parameter dimension.
None of these approaches explain why the strong low-dimensional structure
empirically observed in VQE landscapes can dramatically reduce the effective
difficulty of the optimization.

Our starting point is the geometry of the landscape itself. 
For a periodic ansatz with bounded generators and a bounded Hamiltonian $H$, the
resulting cost is globally Lipschitz and has uniformly bounded curvature.
These regularities arise directly from physically motivated architectures.
Rather than viewing VQE optimization as stochastic descent over a rugged
surface, we treat it as probabilistic elimination carried out over random
geodesic flats of a smooth toroidal manifold.

Low-energy regions in VQE landscapes have been observed to concentrate on
lower-dimensional manifolds or ``dents''~\cite{gard_2020, akande_symmetries, powerQNNs_Abbas_2021}. 
We later describe how strong parameter tying together with architecturally induced locality provides one mechanism that
plausibly yields reduced normal Hessian rank, and we highlight other architectures with similar structure. 
Motivated by this geometry, we introduce \emph{Adaptive Lipschitz Elimination on
the Torus} (\textsc{ALeT}), a modification of classical Lipschitz-elimination
methods~\cite{malherbe2017globaloptimizationlipschitzfunctions} designed to run
in quasipolynomial PAC time under this landscape structure. 
\textsc{ALeT} removes provably suboptimal regions using conservative Lipschitz
certificates and finite-shot confidence intervals, and is efficient when the
normal Hessian rank satisfies $r = O(\log p)$. 
With probability at least $1-\delta$, the returned region is $\varepsilon$-optimal:
every point it contains obeys $C(\theta)\le C^*+\varepsilon$. 
We emphasize that the guarantees established here are conditional on explicit
geometric regularity assumptions (See Section~\ref{sec:prelim}); we do not claim tractability of VQE
optimization in general. 
Furthermore, the solution is $\varepsilon$-optimal for the variational landscape,
not necessarily for the exact ground-state energy of the Hamiltonian.

This framework unifies geometry and probability into a common perspective on
variational optimization. Efficiency arises not from surjectivity or
near-convexity, but from structural properties---periodicity, Lipschitz
continuity, low-dimensional curvature, and fiber regularity---that are often naturally present in
physically motivated settings. 
Section~\ref{sec:prelim} formalizes the setup and assumptions. 
Section~\ref{sec:hessian-rank} identifies architectural regimes admitting low
normal Hessian rank. 
Section~\ref{sec:alet} introduces \textsc{ALeT} and establishes its PAC
guarantees. 
Section~\ref{sec:discussion-outlook} discusses limitations, implications, and
directions for future work.

\paragraph{Contributions.}
This work provides a conditional tractability result for global optimization of VQE landscapes. 
Our main contributions are:

\begin{itemize}
    \item \textsf{Structural regime.}
    We formalize a geometric regime---a Morse--Bott low-energy manifold with low normal curvature dimension and polynomial fiber-regularity---under which approximate global optimization becomes feasible with high probability even in the presence of sampling noise.

    \item \textsf{Noise-robust global optimization.}
    We introduce \textsc{ALeT}, a Lipschitz-based elimination method operating on random $(r{+}1)$-dimensional flats, and prove PAC-style guarantees showing it returns an $\varepsilon$-optimal parameter region using only noisy cost evaluations.

    \item \textsf{Quasi-polynomial complexity.}
    Under the above structural assumptions, \textsc{ALeT} attains sample complexity 
    $\mathrm{poly}(p,\varepsilon^{-1})^{O(r)}$, which becomes quasi-polynomial whenever the normal curvature rank satisfies $r = O(\log p)$.

    \item \textsf{Geometric interpretation.}
    Our analysis indicates that tractability is governed not by global smoothness or convexity surrogates, but by the \emph{effective geometric dimension} of curvature in the low-energy region, providing a precise analogue of structural landscape assumptions in classical nonconvex optimization.
\end{itemize}


\section{Setup}
\label{sec:prelim}

We work with standard product ansatzes on a toroidal parameter space. This
section fixes notation and records the conditional geometric and noise assumptions under
which the main results hold.

\subsection{Parameter space and cost landscape}

The variational circuit is
\[
  U(\theta)\;=\;\prod_{j=1}^{p} e^{-i \theta_j A_j}\,,
  \qquad \theta \in \T^p := (\R/\gamma\Z)^p,
\]
acting on a fixed reference state $\ket{0}$. The cost function is
\[
  \Cost(\theta)\;=\;\langle 0|\, U(\theta)^\dagger H U(\theta)\, |0\rangle,
\]
for a Hamiltonian $H$. We equip $\T^p$ with the geodesic distance
\[
  \dist_{\T^p}(\theta,\theta')
  \;:=\;
  \min_{k\in\Z^p}\,\|\theta-\theta'+\gamma k\|_2,
\]
where $\gamma := 2\pi q$ is the periodicity of the torus. We implicitly identify $\T^p$ with its fundamental domain $[-q\pi,q\pi)^p$
whenever convenient. Closed metric balls are denoted
$\Ball{x}{\mathfrak r}$, and an $\mathfrak r$-net of $S\subseteq\T^p$ is a
finite set of centers whose $\mathfrak r$-balls cover $S$. Hausdorff
$\HH^k$ is the $k$-dimensional volume on $k$-rectifiable sets.

\subsection{Periodic generators and Lipschitz regularity}

The toroidal structure is induced by periodic dependence of $U(\theta)$ on each
parameter. The only structural input is periodicity and a uniform bound on $H$.

\begin{assumption}[Periodic bounded generators]\label{ass:bounded}
Each gate depends on its parameter only through $e^{-i\theta_j A_j}$, with
$A_j$ having rational spectrum $\lambda_k\in\frac{1}{q}\Z$ for some $q\in\mathbb{N}$.
Then $\theta\mapsto U(\theta)$ is $\gamma$-periodic in each coordinate with
$\gamma=2\pi q$. Throughout, we will (\emph{w.l.o.g.}) take $q=1$, and we regard $\theta$ modulo $\gamma$ as living on $\T^p$.
\end{assumption}

\begin{assumption}[Bounded Hamiltonian]\label{ass:H}
The Hamiltonian satisfies $\|H\|_2 \le \Lambda$ for some known $\Lambda>0$.
\end{assumption}

\noindent Under these conditions the cost is globally Lipschitz on $(\T^p,\dist_{\T^p})$.

\begin{proposition}[Global Lipschitz continuity]\label{prop:Lipschitz}
Under Assumptions~\ref{ass:bounded} and~\ref{ass:H},
\[
  |\Cost(\theta)-\Cost(\theta')|
  \;\le\; 2\Lambda\,\sqrt{\sum_j \|A_j\|^2}\,\|\theta-\theta'\|_2
  \qquad \text{for all }\theta,\theta'\in\T^p.
\]
In particular, $\Cost$ is $L$-Lipschitz with
$L\le 2\Lambda\sqrt{\sum_j \|A_j\|^2}$.
\end{proposition}

\noindent The proof is a standard commutator bound for product ansatzes and is included
in App.~\ref{app:L-H-bounds} for completeness. There we also introduce the
\emph{conjugated generators}
\[
  U_{>j}(\theta)\;:=\;\prod_{k=j+1}^p e^{-i\theta_k A_k},
  \qquad
  \widetilde A_j(\theta)\;:=\;U_{>j}^\dagger(\theta)A_j U_{>j}(\theta),
\]
which control both gradients and Hessians in later sections.

\subsection{Shot noise model}

Algorithmically we never observe $C(\theta)$ exactly but only via finite-shot
estimates. We adopt a minimal, homoscedastic model; variance-adaptive
generalizations would only change constants.

\begin{assumption}[Finite-shot measurements]\label{ass:shots}
Each evaluation at $\theta$ returns an unbiased estimator $\wh{\Cost}(\theta)$
computed from $n_{\mathrm{shots}}$ independent quantum measurements, with
outcomes bounded in $[-R/2,R/2]$ for a known $R>0$.
\end{assumption}

\noindent By Hoeffding’s inequality we immediately obtain per-point confidence
intervals.

\begin{lemma}[Energy concentration]\label{lem:energy-ci}
Under Assumption~\ref{ass:shots}, for any fixed $\theta$ and
$\alpha\in(0,1)$,
\[
  \Prb\!\left(|\wh{\Cost}(\theta)-\Cost(\theta)|
    \le \rad(n_\mathrm{shots},\alpha)\right)\;\ge\;1-\alpha, 
\]
where
\[
  \rad(n_\mathrm{shots},\alpha)
  := R\sqrt{\frac{\ln(2/\alpha)}{2n_\mathrm{shots}}}.
\]
\end{lemma}

\noindent These radii are fed directly into the elimination certificates in
Section~\ref{sec:alet}.

\subsection{Dent manifold and Morse--Bott structure}

The PAC guarantees are localized to a geometrically regular low-energy region
containing the global minimum. We model this region as a Morse--Bott critical
submanifold.

\begin{definition}[Dent manifold]\label{def:dent}
A subset $M\subset\T^p$ is a dent manifold for $C$ if,
\begin{enumerate}
  \item \emph{Global minimum containment:} there exists $\theta^\star\in M$
    such that $C(\theta^\star)\le C(\theta)$ for all $\theta\in\T^p$.
  \item \emph{Criticality on $M$:} $\nabla C(\theta)=0$ for all $\theta\in M$.
  \item \emph{Constant normal rank:} there exists $r\ge 1$ s.t.
    $\mathrm{rank}\bigl(\nabla^2 C(\theta)\bigr)
      =\mathrm{rank}\bigl(\nabla^2 C(\theta)\!\mid_{N_\theta M}\bigr)
      = r$
    for all $\theta\in M$.
  \item \emph{Positive curvature in normal directions:}
    $\nabla^2C(\theta)\!\mid_{N_\theta M}\succ 0$ for all $\theta\in M$.
\end{enumerate}
Here $N_\theta M$ is the normal space to $M$ at $\theta$.
\end{definition}

Intuitively, $M$ is a smooth, low-dimensional ``valley floor'' on which the
cost is flat. This is the precise version of the ``dent'' regime motivated in the introduction. Parameter symmetries, layer tying, and gauge freedoms all
naturally generate such Morse--Bott regions. Approximate Morse--Bott would suffice; we state the exact version for clarity and ease of analysis.

\subsection{Fiber regularity of the dent}

The randomized slicing analysis in Section~\ref{sec:alet} requires that $M$ not
concentrate an exponential fraction of its volume on a vanishing set of
fibers. We encode this by a polynomial fiber-regularity condition.

Fix a coordinate-aligned, linear $(r{+}1)$-plane $V\subset\R^p$ with orthogonal
complement $V^\perp$ of dimension $m{-}1=p{-}r{-}1$, and let $\T_{V^\perp}$ denote the subtorus corresponding to $V^\perp$. For $b\in\T_{V^\perp}$,
write
\[
  X(b) \;:=\; \HH^1\!\left(M\cap E(b,V)\right),
\]
where $\HH^1$ denotes $1$-dimensional Hausdorff measure (so $X(b)=0$ whenever
the intersection is empty or $0$-dimensional), and $E(b, V) \subset \Tp$ is the geodesic flat through $b$ parallel to $V$. Let $\mu := \EE_b[X(b)]$ denote
the expectation with respect to normalized Haar measure on $\T_{V^\perp}$.

\begin{assumption}[Polynomial fiber regularity]\label{ass:fiber-regularity}
We say $M$ has polynomial fiber regularity (with respect to $V$) if there exists
a constant $A_{p,m}\ge 1$, depending at most polynomially on $(p,m)$, such that
\[
   X(b) \;\le\; A_{p,m}\,\mu
   \qquad\text{for Haar-almost every } b\in\T_{V^\perp}.
\]
\end{assumption}
\noindent Therefore no typical slice exceeds the mean slice length by more than a
polynomial factor in $(p,m)$.

\begin{example}[Tubular dent intersection implies fiber regularity]
\label{ex:tube-implies-fiber}
Suppose that, for $(b,t)\in \T_{V^\perp}\times\R$, the set
$M\cap E(b, V)$ is contained in a tube
\[
   M \cap E(b, V)\;\subseteq\; \{(b,t): b\in B,\; -\rho(b)\le t \le \rho(b)\},
\]
where $B\subset\T_{V^\perp}$ is measurable and
\[
  0 < \rho_{\min}\le \rho(b)\le \rho_{\max}\qquad(b\in B),
  \qquad
  \frac{\rho_{\max}}{\rho_{\min}} \le \poly(p,m),
\]
and where the base has non-negligible measure,
\[
   \frac{\HH^{m-1}(B)}{(2\pi)^{m-1}} \;\ge\; \frac{1}{\poly(p,m)}.
\]
Then $M$ satisfies Assumption~\ref{ass:fiber-regularity} with
$A_{p,m} = \poly(p,m)$.
\end{example}

\begin{proof}
For $b\in B$, the slice $M\cap E(b,V)$ coincides with an
interval of length $X(b)$, and for $b\notin B$ we have $X(b)=0$. Thus
\[
  \mu
  = \EE_b X(b)
  = \frac{1}{(2\pi)^{m-1}}\int_{\T_{V^\perp}}X(b)\,db
  \geq \frac{2\rho_\mathrm{min}}{\poly(p,m)}.
\]
For $b\in B$,
\[
  X(b)\le 2\rho_{\max}
  \le \frac{\rho_{\max}}{\rho_{\min}}
       \frac{(2\pi)^{m-1}}{\HH^{m-1}(B)}\,\mu
  = A_{p,m}\,\mu,
\]
and for $b\notin B$ the inequality is trivial since $X(b)=0$. Hence
Assumption~\ref{ass:fiber-regularity} holds with the stated $A_{p,m}$.
\end{proof}

\begin{remark}\label{rmk:fiber-regularity}
Assumption~\ref{ass:fiber-regularity} is a structural regularity condition placed on the dent manifold $M$. It rules out the pathological situation where $M$ concentrates an exponential fraction of its volume on a vanishing set of fibers of the projection $\pi_{V^\perp}$. The assumption is not derived from the curvature or positive-reach hypotheses; it is an additional geometric requirement that determines whether randomized slicing can be used to certify non-emptiness with polynomial overhead.

In highly structured VQE landscapes, empirical studies of dent-like minima and symmetry-reduced low-energy regions suggest that the mass of near-optimal sets is typically spread over many fibers rather than collapsing onto a negligible subset. Assumption~\ref{ass:fiber-regularity} should be viewed as a conditional geometric hypothesis, such that whenever a VQE landscape exhibits a Morse--Bott minimum whose low-energy region is not exponentially concentrated along any coordinate fiber, the guarantees of Section 4 go through.

Establishing verifiable sufficient conditions for fiber regularity--possibly using curvature measures~\cite{federer1959curvature, RatajZaehle2019} or projection identities of Federer-Kubota type~\cite{lotz2019concentrationintrinsicvolumesconvex}--is an interesting direction but not required for our results. Though harder to motivate, a simple and weaker second-moment bound also suffices for the Paley-Zygmund step of Section~\ref{sec:alet}.
\end{remark}


\section{Low Hessian Rank}
\label{sec:hessian-rank}

The efficiency of \textsc{ALeT} depends on curvature concentrating in a small number
of normal directions. Crucially, the relevant Hessian is the \emph{normal
Hessian along the dent manifold} $M$, not the ambient Hessian at generic
parameter values. 

In this section we describe architectural conditions under
which the normal curvature rank satisfies $r = O(\log p)$. Parameter tying alone does not guarantee low curvature rank unless such tying is so strong so as to induce a rank $O(\log p)$ tying map (chain rule); rather,
tying must interact with locality and multiscale structure to collapse the
span of \emph{backpropagated generator templates} that control the curvature
near $M$. The following discussion and architectural examples serve as motivation, illustrating mechanisms that plausibly yield 
low effective curvature rank, but we do not prove low-rank bounds from architectural assumptions in the current work.

\subsection{Tying, locality, and the effective curvature dimension}

Recall from App.~\ref{app:linearansatz} that the Hessian in a direction $v \in \mathbb{R}^p$ can be written as
\[
  \partial^2_v C(\theta)
  = -\langle \lbrack K_\theta(v), [H, K_\theta(v)] \rbrack\rangle + \langle i[H, \dot K]\rangle ,
\]
where $K_\theta(v) = \sum_j v_j \tilde A_j(\theta)$ is a linear combination of
the \emph{backpropagated/conjugated generators}
$\tilde A_j(\theta) = U_{>j}^\dagger(\theta) A_j U_{>j}(\theta)$, so that the curvature rank at $\theta$ is controlled by the dimension of
\[
  \mathcal{K}(\theta)
  := \mathrm{span}\{\tilde A_j(\theta): 1\le j\le p\},
\]
and its restriction to the normal bundle $N_\theta M$.
If $\dim \mathcal{K}(\theta) = r_{\mathrm{eff}}$, then
$\mathrm{rank}(\nabla^2 C(\theta)|_{N_\theta M}) \le r_{\mathrm{eff}}$.

Parameter tying enters only indirectly: tying groups parameters so that many
$\tilde A_j(\theta)$ coincide \emph{once propagated through the circuit}.  
Whether this collapse is strong enough to force $r_{\mathrm{eff}} = O(\log p)$
depends on the circuit architecture.

\begin{definition}[Operator-template collapse]
We say an ansatz exhibits operator-template collapse of order $r$ near $M$ if
there exists a neighborhood $\mathcal{U}$ of $M$ such that
\[
  \dim \mathcal{K}(\theta) \le r
  \qquad(\theta \in \mathcal{U}).
\]
\end{definition}

This is the actual structural condition needed for low curvature rank. It is a
restriction on the \emph{local effective degrees of freedom of the generator
algebra}.

\subsection{Tied parameter classes and template collapse}

Let the $p$ parameters be partitioned into $k$ tied classes, with $k = O(\log p)$.
Write $T: \mathbb{R}^p \to \mathbb{R}^k$ for the induced tying map.
Tying alone restricts directions in parameter space, but, does not, except under the aforementioned condition, by itself
limit $\mathcal{K}(\theta)$.

However, when tying is combined with \emph{locality} or \emph{multiscale
structure}, the backpropagated operators $\tilde A_j(\theta)$ do not proliferate
independently, but rather remain confined to a light cone of bounded width or a
constant-size template at each scale. The effect is that many distinct
parameters yield identical or linearly dependent $\tilde A_j(\theta)$ once
pushed through the circuit.

Consider a depth-$D$ architecture on a bounded-degree lattice, with at most $O(1)$
tied classes per layer. If the backward light cone of each gate has size $O(1)$
(independent of $p$), then for $\theta$ in a neighborhood of any Morse--Bott
minimum, $\dim \mathcal{K}(\theta) = O(D).$

Locality ensures that backpropagation maps each tied class to one of $O(D)$ operator templates, and tangential directions along $M$ contribute no curvature. Choosing depth $D = \Theta(\log p)$ gives $\dim \mathcal{K}(\theta) = O(\log p)$, so the normal Hessian rank is at most $O(\log p)$.

\subsection{Architectural regimes achieving $r=O(\log p)$}

\paragraph{Hardware-efficient layers with tied gate types.~\cite{Kandala_2017, Leone2024practicalusefulness}}
In $D$ alternating layers of local rotations and entanglers on a bounded-degree
graph, tying all gates of the same type in each layer gives $O(1)$ tied classes
per layer.
\[
  r_{\mathrm{eff}} = O(D).
\]
With $D = \Theta(\log p)$, $r_{\mathrm{eff}} = O(\log p)$.

\paragraph{Shallow multiscale ansatzes (TTN/MERA).~\cite{Vidal_2008, Shi_2006, Seitz2023simulatingquantum}}
In tree or MERA layouts, the number of scales is $\Theta(\log p)$, and tying
parameters within each scale gives $O(1)$ tied classes per scale. Each scale
induces only a constant number of backpropagated templates, suggesting
$r_{\mathrm{eff}} = O(\log p)$ in a neighborhood of the optimum.

\section{Adaptive Lipschitz Elimination on the Torus (\textsc{ALeT})}
\label{sec:alet}

\textsc{ALeT} adapts the Lipschitz-elimination of
LIPO~\cite{malherbe2017globaloptimizationlipschitzfunctions} to the geometric
setting of VQE landscapes on the torus. The algorithm differs in two key respects.
First, instead of working in the full parameter space $\T^p$, it restricts
search to random $(r{+}1)$-dimensional geodesic flats $E\subset\T^p$ that
intersect the Morse--Bott dent $M$ with positive probability. Second, it
replaces exact function evaluations with finite-shot estimates and embeds
confidence intervals directly into the elimination rule. These modifications
make elimination feasible in the low--normal-rank regime $r=O(\log p)$ and
yield PAC guarantees under shot noise.

Let $M\subset\T^p$ be the dent manifold with intrinsic dimension $m$ and normal
Hessian rank $r=p-m$. Along $M$, the cost is locally flat ($C|_M\equiv C^\star$
on connected components), while the normal directions exhibit strictly positive
curvature.

\subsection{Randomized slicing of the dent}
\label{subsec:rand-slicing}

We parametrize flats via a linear $(r{+}1)$-plane $V\in\Gr(r{+}1,\R^p)$ and a
translation $v\in\T^p$. Writing $w:\R^p\to\T^p$ for the coordinate-wise
mod-$2\pi$ map, define the geodesic flat
\[
  E(v,V)\;:=\;w(v+V)\ \subset\ \T^p.
\]
For each such flat we consider the one-dimensional slice length of the
intersection with the dent,
\[
  X(v)\ :=\ \HH^1\bigl(M\cap E(v,V)\bigr)\ \in [0,\infty).
\]

The following identity expresses the average slice length in terms
of the $m$-dimensional volume of $M$, $\HH^m (M)$, and the average transversality of $M$ to
$V^\perp$, $K(V; M)$.

\begin{proposition}[Translation average]
\label{prop:full-torus-expectation}
Fix a coordinate-aligned $(r{+}1)$-plane $V\subset\R^p$ and sample
$v\sim\Unif(\T^p)$. Then
\[
  \EE_v X(v)
  \;=\;
  (2\pi)^{-(m-1)}\,K(V;M)\,\HH^m(M),
\]
where
\[
  K(V;M)
  :=
  \frac{1}{\HH^m(M)}
  \int_M
    J_{m-1}\!\bigl(D\pi_{V^\perp}\!\mid_{T_xM}\bigr)\,d\HH^m(x)
\]
and $J_{m-1}$ is the $(m{-}1)$-Jacobian of $\pi_{V^\perp}$ restricted to
$T_xM$.
\end{proposition}

\begin{proof}
Decompose $\R^p = V\oplus V^\perp$ so that
$\T^p\simeq\T^{\,r+1}\times\T^{\,m-1}$, and let
$\pi_{V^\perp}:\T^p\to\T^{\,m-1}$ be the projection. For $v=(a,b)$ with
$a\in\T^{\,r+1}$ and $b\in\T^{\,m-1}$, the flat $E(v,V)$ is precisely the fiber
$\pi_{V^\perp}^{-1}(b)$. Fubini and the coarea formula applied to
$\pi_{V^\perp}\!\mid_M$ give
\[
  \int_{\T^p} X(v)\,d\HH^p(v)
  =
  (2\pi)^{r+1}
  \int_M J_{m-1}\!\bigl(D\pi_{V^\perp}\!\mid_{T_xM}\bigr)\,d\HH^m(x),
\]
and dividing by $\HH^p(\T^p)=(2\pi)^p$ yields the stated expectation.
\end{proof}

Under Assumption~\ref{ass:fiber-regularity}, the slice lengths are not allowed
to concentrate exponentially on a negligible set of fibers.

\begin{lemma}[Second moment under fiber regularity]
\label{lem:second-moment-fiber}
Under Assumption~\ref{ass:fiber-regularity},
\[
  \EE_{v} X(v)^2
  \;\le\;
  \frac{A_{p,m}}{(2\pi)^{2(m-1)}}\,K(V;M)^2\,\HH^m(M)^2.
\]
\end{lemma}

\begin{proof}
Assumption~\ref{ass:fiber-regularity} states that
$X(v)\le A_{p,m}\mu$ almost everywhere, where
$\mu=\EE_v X(v)$. Thus $X(v)^2\le A_{p,m}\mu\,X(v)$ and
\[
  \EE_v X(v)^2
  \;\le\;
  A_{p,m}\mu^2.
\]
Substituting $\mu$ from Proposition~\ref{prop:full-torus-expectation} gives the
claim.
\end{proof}

Combining the first and second moments via Paley--Zygmund yields a uniform
lower bound on the probability that a random flat intersects $M$.

\begin{theorem}[Hit probability for a random flat]
\label{thm:hit-PZ}
Let $v\sim\Unif(\T^p)$ and $X(v)$ as above. Then
\[
  \Pr_v\!\big[M\cap E(v,V)\neq\varnothing\big]
  \;\ge\;
  \frac{\bigl(\EE_v X(v)\bigr)^2}{\EE_v X(v)^2}
  \;\ge\;
  A_{p,m}^{-1}.
\]
\end{theorem}

\begin{remark}
For a generic choice of $(r{+}1)$-plane $V$, standard transversality arguments
imply $K(V;M)>0$ for any fixed Morse--Bott $M$; we restrict to such $V$. Assumption~\ref{ass:fiber-regularity}
then ensures that the second moment of $X(v)$ is controlled polynomially in
$(p,m)$. Without such a regularity condition, one can still use positive reach
and translative intersection formulas~\cite{federer1959curvature,RatajZaehle2019,asymptoticgeometricanalysis1},
but the resulting bounds generally deteriorate exponentially in $m$ when a large
fraction of the mass of $M$ is concentrated in a small set of fibers.
\end{remark}

\begin{figure}[t]
    \centering
    \includegraphics[width=0.7\linewidth]{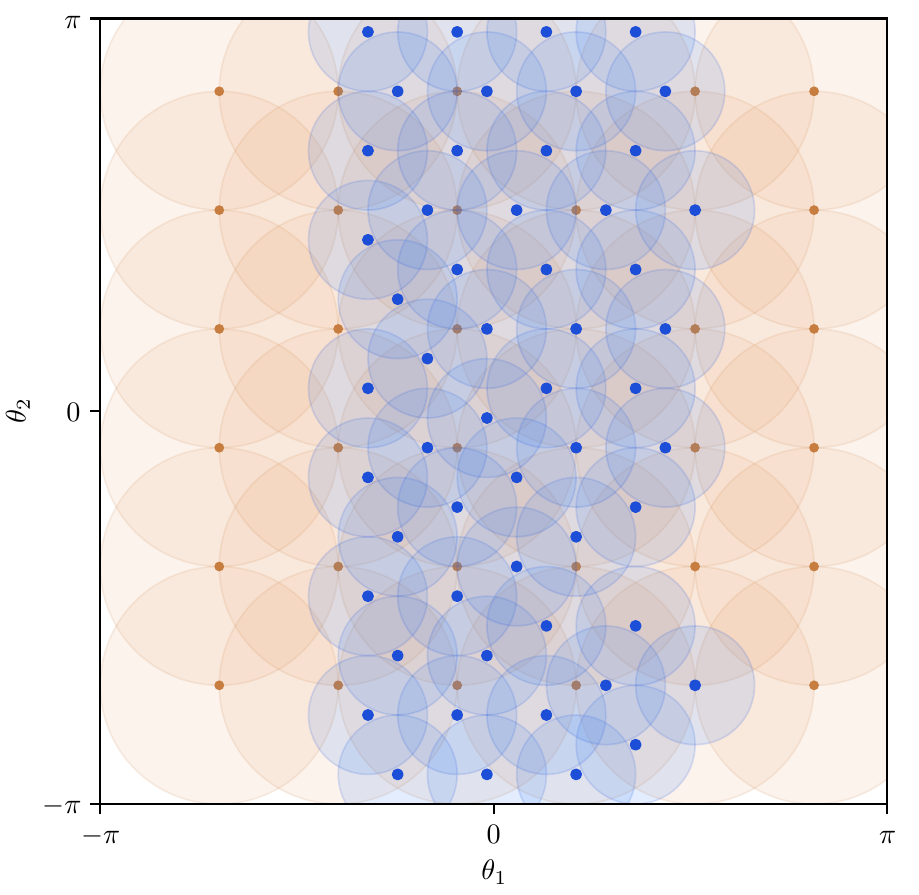}
    \caption{Two steps of \textsc{ALeT}: (orange) an initial $\mathfrak r$-net eliminates exterior balls via the EC; (blue) a refined $\mathfrak r/2$-net covers the surviving set.}
    \label{fig:alet}
\end{figure}

\subsection{Noisy elimination certificates on a flat}

Fix a flat $E=E(v,V)$ that intersects $M$, and endow $E$ with its geodesic
metric inherited from $\T^p$. Let $L>0$ be any valid global Lipschitz constant
for $C$ on $\T^p$. For a subset $S\subseteq E$ and radius $\mathfrak r>0$, write
$\mathcal C_{\mathfrak r}(S)$ for an $\mathfrak r$-net of $S$.

On round $t$, \textsc{ALeT} maintains a survivor set $S_t\subseteq E$ and an
$\mathfrak r_t$-net $\mathcal C_t:=\mathcal C_{\mathfrak r_t}(S_t)$, with
$\mathfrak r_{t+1}=\mathfrak r_t/\Delta$ for some fixed $\Delta>1$. Each query
at $\theta\in\mathcal C_t$ returns an unbiased estimator $\widehat C_t(\theta)$
from $n_{\mathrm{shots}}$ independent shots, bounded in $[-R/2,R/2]$ by
Assumption~\ref{ass:shots}. For any confidence level $\alpha\in(0,1)$ define
\[
  \rad(\alpha)
  := R\sqrt{\frac{\ln(2/\alpha)}{2\,n_{\mathrm{shots}}}}.
\]
With $N_t:=|\mathcal C_t|$, we set
\[
  \alpha_{t,\theta}
  := \frac{6\delta_{\mathrm{noise}}}{\pi^2\,t^2\,N_t},
  \qquad \theta\in\mathcal C_t,
\]
and form
\[
  \mathrm{LCB}_t(\theta)
  := \widehat C_t(\theta)-\rad(\alpha_{t,\theta}),\qquad
  \mathrm{UCB}_t^\star
  := \min_{\phi\in\mathcal C_t}
      \bigl(\widehat C_t(\phi)+\rad(\alpha_{t,\phi})\bigr).
\]
The \emph{noisy elimination certificate} at round $t$ is
\begin{equation}
  \mathrm{LCB}_t(\theta)
  \;>\;
  \mathrm{UCB}_t^\star + L\,\mathfrak r_t.
  \label{eq:EC}
\end{equation}
If~\eqref{eq:EC} holds, then for every $x\in B_E(\theta,\mathfrak r_t)$ we have
$C(x)>\mathrm{UCB}_t^\star \geq \min_{\phi\in\mathcal C_t}C(\phi)$, hence the entire ball
is certifiably suboptimal with respect to confidence $\alpha_{t, \theta}$ and can be removed. Our choice of $\alpha_t$ is not strictly necessary, but simplifies analysis. In practice, it may often be useful to maintain a constant $\rad(\alpha_t)$ and include this in the chosen $\varepsilon$; for such cases, note that the growth of required $n_\mathrm{shots}$ at timestep $t$, $n_t$, is not prohibitive under the adaptive scheduler.

\subsection{\textsc{ALeT} on a fixed flat}

\begin{algorithm}[H]
\caption{\textsc{ALeT} on a fixed flat $E=E(v,V)\subset\T^p$}
\label{alg:ALeT}
\begin{algorithmic}[1]
\State \textbf{Inputs:} final radius $\mathfrak r_{\mathrm{fin}}$, shrink factor $\Delta>1$, Lipschitz bound $L$, noise budget $\delta_{\mathrm{noise}}\in(0,1)$.
\State Initialize $S_0\gets E$, choose $\mathfrak r_0$, and set $\mathcal C_0\gets \mathcal C_{\mathfrak r_0}(S_0)$.
\For{$t=0,1,2,\dots$}
  \State Query $\widehat C_t(\theta)$ for all $\theta\in\mathcal C_t$.
  \State Compute $\mathrm{LCB}_t(\theta)$ and $\mathrm{UCB}_t^\star$ as above.
  \State \textbf{Elimination Step:}
    \[
      \mathcal C_t^{\mathrm{keep}}
      \gets
      \Bigl\{
        \theta\in\mathcal C_t:\;
        \mathrm{LCB}_t(\theta)\le
        \mathrm{UCB}_t^\star + L\mathfrak r_t
      \Bigr\}.
    \]
  \State Set
    $S_{t+1}\gets \bigcup_{\theta\in\mathcal C_t^{\mathrm{keep}}} B_E(\theta,\mathfrak r_t)$.
  \If{$\mathfrak r_t\le \mathfrak r_{\mathrm{fin}}$}
    \State \textbf{break}
  \EndIf
  \State $\mathfrak r_{t+1}\gets \mathfrak r_t/\Delta$,\quad
         $\mathcal C_{t+1}\gets \mathcal C_{\mathfrak r_{t+1}}(S_{t+1})$.
\EndFor
\State \textbf{Output:} $S_T$ as an $\varepsilon$-optimal region on $E$.
\end{algorithmic}
\end{algorithm}

\subsection{Correctness guarantees on a single flat}
\label{subsec:alet-guarantees}

We first note that survivor sets are nested.

\begin{theorem}[Monotonicity on a flat]
\label{thm:monotonicity}
For Algorithm~\ref{alg:ALeT} one has $S_{t+1}\subseteq S_t$ for all $t$.
\end{theorem}

\begin{proof}
By construction, $S_{t+1}$ is the union of balls around $\mathcal C_t^{\mathrm{keep}}\subseteq\mathcal C_t$, so every kept ball at round $t{+}1$ is contained in a ball that survived round $t$.
\end{proof}

\noindent The main guarantee is an $\varepsilon$-optimality statement on $E$ with high
probability over the measurement noise. Let $\mathfrak r_T$ be the final
radius reached and define
\[
  \varepsilon_T
  := 3L\,\mathfrak r_T + 4\,\rad(\alpha_{T,\theta}).
\]

\begin{lemma}[Survival of the minimizer on a flat]
\label{lem:minimizer-survival}
Let $E\subset\T^p$ be a fixed flat and let $x^\star\in E$ be a minimizer of $C$ over $E$.
On the event $\mathcal G$ (all confidence intervals hold), the point $x^\star$
lies in every survivor set $S_t$ and, for each round $t$, there exists a kept
center $\theta_t\in\mathcal C_t^{\mathrm{keep}}$ such that
\[
  \dist_E(\theta_t,x^\star) \le \mathfrak r_t.
\]
\end{lemma}

\begin{proof}
We argue by induction on $t$. For $t=0$ we have $S_0=E$, so $x^\star\in S_0$.
Assume $x^\star\in S_t$. Because $\mathcal C_t$ is an $\mathfrak r_t$-net of
$S_t$, there exists $\zeta_t\in\mathcal C_t$ with
$\dist_E(\zeta_t,x^\star)\le\mathfrak r_t$.

We first show that $\zeta_t$ cannot be eliminated on $\mathcal G$. Suppose, for
contradiction, that $\zeta_t\notin \mathcal C_t^{\mathrm{keep}}$, i.e.\ the
elimination certificate~\eqref{eq:EC} holds at $\zeta_t$:
\[
  \mathrm{LCB}_t(\zeta_t)
  \;>\;
  \mathrm{UCB}_t^\star + L\,\mathfrak r_t.
\]
On $\mathcal G$, we have
$C(\zeta_t)\ge \mathrm{LCB}_t(\zeta_t)$ and
$C(\zeta_t)\le C(x^\star)+L\mathfrak r_t$ by Lipschitz continuity and
$\dist_E(\zeta_t,x^\star)\le\mathfrak r_t$. Thus
\[
  C(x^\star)+L\mathfrak r_t
  \;\ge\;
  C(\zeta_t)
  \;\ge\;
  \mathrm{LCB}_t(\zeta_t)
  \;>\;
  \mathrm{UCB}_t^\star + L\,\mathfrak r_t,
\]
which implies $C(x^\star)>\mathrm{UCB}_t^\star$. On the other hand, for any
$\phi\in\mathcal C_t$, the event $\mathcal G$ gives
$C(\phi)\le\widehat C_t(\phi)+\rad(\alpha_{t,\phi})$, so
\[
  \mathrm{UCB}_t^\star
  = \min_{\phi\in\mathcal C_t}
    \bigl(\widehat C_t(\phi)+\rad(\alpha_{t,\phi})\bigr)
  \;\ge\;
  \min_{\phi\in\mathcal C_t} C(\phi)
  \;\ge\;
  C(x^\star),
\]
since $x^\star$ minimizes $C$ over $E\supseteq\mathcal C_t$. This contradicts
$C(x^\star)>\mathrm{UCB}_t^\star$. Hence $\zeta_t\in\mathcal C_t^{\mathrm{keep}}$.

By definition of $S_{t+1}$,
\[
  S_{t+1}
  = \bigcup_{\theta\in\mathcal C_t^{\mathrm{keep}}} B_E(\theta,\mathfrak r_t),
\]
and $x^\star\in B_E(\zeta_t,\mathfrak r_t)$ by construction, so
$x^\star\in S_{t+1}$. Taking $\theta_t:=\zeta_t$ completes the induction and
the proof.
\end{proof}

\begin{theorem}[$\varepsilon$-optimality with confidence]
\label{thm:optimality}
Let $C^\star:=\min_{\theta\in E} C(\theta)$. With probability at least
$1-\delta_{\mathrm{noise}}$ (over all measurement noise),
every $x\in S_T$ returned by Algorithm~\ref{alg:ALeT} satisfies
\[
  C(x)\ \le\ C^\star + \varepsilon_T.
\]
\end{theorem}
\begin{proof}
Let $\mathcal G$ be the event that every queried $\theta$ at every round has its
true value within the corresponding confidence interval, i.e.
\[
  |\,\widehat C_t(\theta)-C(\theta)\,|\le \rad(\alpha_{t,\theta})
  \quad\text{for all $t$ and all $\theta\in\mathcal C_t$}.
\]
By a union bound and the choice of $\alpha_{t,\theta}$,
\[
  \Pr(\mathcal G^c)
  \;\le\;
  \sum_{t\ge 1}\sum_{\theta\in\mathcal C_t}\alpha_{t,\theta}
  \;\le\;
  \sum_{t\ge 1}\frac{6\delta_{\mathrm{noise}}}{\pi^2 t^2}
  \;\le\;
  \delta_{\mathrm{noise}}.
\]
Hence $\Pr(\mathcal G)\ge 1-\delta_{\mathrm{noise}}$.

Fix a round $t$ and work on the event $\mathcal G$. Let
$R_t := \max_{\phi\in\mathcal C_t}\rad(\alpha_{t,\phi})$; note that in our
construction $\rad(\alpha_{t,\phi})$ does not depend on $\phi$, so
$\rad(\alpha_{t,\phi}) = R_t$ for all $\phi\in\mathcal C_t$.

\medskip
Take any kept center
$\theta\in\mathcal C_t^{\mathrm{keep}}$. Since $\theta$ is not eliminated, the
elimination certificate~\eqref{eq:EC} fails at $\theta$, so
\[
  \mathrm{LCB}_t(\theta)
  \;=\; \widehat C_t(\theta)-\rad(\alpha_{t,\theta})
  \;\le\;
  \mathrm{UCB}_t^\star + L\,\mathfrak r_t.
\]
Rearranging and using $\rad(\alpha_{t,\theta}) = R_t$ gives
\[
  \widehat C_t(\theta)
  \;\le\;
  \mathrm{UCB}_t^\star + L\,\mathfrak r_t + R_t.
\]
On $\mathcal G$ we have
$C(\theta)\le \widehat C_t(\theta)+\rad(\alpha_{t,\theta})
 = \widehat C_t(\theta)+R_t$, hence
\begin{equation}
  C(\theta)
  \;\le\;
  \mathrm{UCB}_t^\star + L\,\mathfrak r_t + 2R_t.
  \label{eq:kept-center-bound}
\end{equation}

\medskip\noindent
Let $x^\star\in E$ be a minimizer of $C$ over $E$, so $C(x^\star)=C^\star$.
Because $\mathcal C_t$ is an $\mathfrak r_t$-net of $S_t$ and
$x^\star\in S_t$, there exists $\zeta\in\mathcal C_t$ such that
$\dist_E(\zeta,x^\star)\le\mathfrak r_t$. Lipschitz continuity then implies
\[
  C(\zeta)
  \;\le\;
  C(x^\star) + L\,\mathfrak r_t
  \;=\;
  C^\star + L\,\mathfrak r_t.
\]
On $\mathcal G$,
\[
  \widehat C_t(\zeta) + \rad(\alpha_{t,\zeta})
  \;\le\;
  C(\zeta) + 2\rad(\alpha_{t,\zeta})
  \;\le\;
  C^\star + L\,\mathfrak r_t + 2R_t.
\]
By definition of $\mathrm{UCB}_t^\star$,
\[
  \mathrm{UCB}_t^\star
  \;=\;
  \min_{\phi\in\mathcal C_t}
  \bigl(\widehat C_t(\phi)+\rad(\alpha_{t,\phi})\bigr)
  \;\le\;
  \widehat C_t(\zeta)+\rad(\alpha_{t,\zeta}),
\]
so combining the last two displays yields
\begin{equation}
  \mathrm{UCB}_t^\star
  \;\le\;
  C^\star + L\,\mathfrak r_t + 2R_t.
  \label{eq:ucb-vs-opt}
\end{equation}

\medskip\noindent
Substituting~\eqref{eq:ucb-vs-opt} into
\eqref{eq:kept-center-bound} gives, for every
$\theta\in\mathcal C_t^{\mathrm{keep}}$,
\[
  C(\theta)
  \;\le\;
  \bigl(C^\star + L\,\mathfrak r_t + 2R_t\bigr)
  + L\,\mathfrak r_t + 2R_t
  \;=\;
  C^\star + 2L\,\mathfrak r_t + 4R_t.
\]

\noindent
At the final round $T$, every $x\in S_T$ lies within distance at most
$\mathfrak r_T$ of some kept center $\theta_x\in\mathcal C_T^{\mathrm{keep}}$ by
construction of the net. Using the Lipschitz property once more,
\[
  C(x)
  \;\le\;
  C(\theta_x) + L\,\mathfrak r_T
  \;\le\;
  C^\star + 3L\,\mathfrak r_T + 4R_T,
\]
where $R_T = \max_{\phi\in\mathcal C_T}\rad(\alpha_{T,\phi})
 = \rad(\alpha_{T,\theta^\star})$ for any $\theta^\star\in\mathcal C_T$.

By the definition of $\varepsilon_T$,
$\varepsilon_T = 3L\,\mathfrak r_T + 4\,\rad(\alpha_{T,\theta}) = 3L\,\mathfrak r_T + 4R_T$,
so we obtain
\[
  C(x)\ \le\ C^\star + \varepsilon_T
\]
for all $x\in S_T$ on the event $\mathcal G$. Since
$\Pr(\mathcal G)\ge 1-\delta_{\mathrm{noise}}$, the theorem follows.
\end{proof}

We now quantify the number of oracle calls required to reach a prescribed final
radius on a single flat.

\begin{theorem}[Sample complexity on a single flat]
\label{thm:flat-sample-complexity}
Let $E=E(v,V)\subset\T^p$ be an $(r{+}1)$-dimensional geodesic flat and assume
$C|_E$ is $L$-Lipschitz. Let $\varepsilon_{\mathrm{eff}}:=\varepsilon_T-4\rad_T>0$ be the effective accuracy at the final radius
\[
  \mathfrak r_T
  = \frac{\varepsilon_{\mathrm{eff}}}{3L},
  \qquad
  \rad_T:=\rad(n_{\mathrm{shots}},\alpha_{T,\theta}).
\]
Let $\kappa_{r+1}:=\HH^{\,r+1}(B_1(0))$ denote the volume of the $(r{+}1)$-dimensional
Euclidean unit ball. Then there exists a universal constant $C_0>0$ such that
\[
  N_{\mathrm{flat}}
  \;\le\;
  C_0\,\kappa_{r+1}^{-1}
  \left(\frac{6\pi L}{\varepsilon_{\mathrm{eff}}}\right)^{r+1},
\]
where $N_{\mathrm{flat}}$ is the total number of oracle queries made by
Algorithm~\ref{alg:ALeT} on $E$.
\end{theorem}

\begin{proof}
Each connected component of $E$ is a flat $(r{+}1)$-torus of diameter at most
$2\pi$ in the induced metric. Let $d:=r{+}1$. For any $\rho\in(0,2\pi]$, a
standard volume argument yields
\begin{equation}
  N_{\mathrm{cov}}(E,\rho)
  \;\le\;
  C_1\,\frac{\HH^d(E)}{\HH^d(B_\rho)}
  \;=\;
  C_1\,\kappa_d^{-1}\Bigl(\frac{2\pi}{\rho}\Bigr)^d,
  \label{eq:covering-number}
\end{equation}
for some universal $C_1>0$. At the last elimination round, $\mathcal C_T$ is an
$\mathfrak r_T$-net of a subset of $E$, so
\[
  |\mathcal C_T|
  \;\le\;
  N_{\mathrm{cov}}(E,\mathfrak r_T)
  \;\le\;
  C_1\,\kappa_d^{-1}
  \Bigl(\frac{2\pi}{\mathfrak r_T}\Bigr)^d.
\]
Because $\mathfrak r_{t+1}=\mathfrak r_t/\Delta$, the covering numbers
$|\mathcal C_t|$ grow geometrically in $t$, and the total number of queries
satisfies
\[
  N_{\mathrm{flat}}
  = \sum_{t=0}^T |\mathcal C_t|
  \;\le\;
  C_2\,|\mathcal C_T|
  \;\le\;
  C_0\,\kappa_d^{-1}
  \Bigl(\frac{2\pi}{\mathfrak r_T}\Bigr)^d
\]
for a universal constant $C_2>0$ and $C_0:=C_1C_2$. Substituting
$\mathfrak r_T=\varepsilon_{\mathrm{eff}}/(3L)$ yields the stated bound.
\end{proof}

The dependence on dimension can be made explicit via the asymptotics of
$\kappa_d$.

\begin{corollary}[Quasi-polynomial scaling for $r=O(\log p)$]
\label{cor:quasipoly-flat}
Under the hypotheses of Theorem~\ref{thm:flat-sample-complexity}, one has
\[
  \log N_{\mathrm{flat}}
  \;=\;
  O\!\left(
    d\log d + d\log\frac{L}{\varepsilon_{\mathrm{eff}}}
  \right),
  \qquad d=r{+}1.
\]
In particular, if $r=O(\log p)$, then $d=O(\log p)$ and
\[
  N_{\mathrm{flat}}
  \;\le\;
  \exp\!\Bigl(
    O\bigl(\log p\,[\,\log\log p + \log(L/\varepsilon_{\mathrm{eff}})\,]\bigr)
  \Bigr),
\]
i.e.\ quasi-polynomial in both $p$ and $1/\varepsilon_{\mathrm{eff}}$.
\end{corollary}

\begin{proof}
The unit-ball volume satisfies
\[
  \kappa_d
  = \frac{\pi^{d/2}}{\Gamma(d/2+1)}
  = \exp\!\bigl(-\tfrac d2\log d + O(d)\bigr),
\]
so $\log\kappa_d^{-1}=O(d\log d)$ as $d\to\infty$. Taking logs in
Theorem~\ref{thm:flat-sample-complexity} gives
\[
  \log N_{\mathrm{flat}}
  \;\le\;
  O\bigl(d\log d\bigr)
  + d\log\frac{L}{\varepsilon_{\mathrm{eff}}},
\]
which is the first claim. If $r=O(\log p)$ then $d=O(\log p)$ and
$\log d = O(\log\log p)$, yielding the stated quasi-polynomial bound.
\end{proof}
\subsection{Global PAC guarantee over multiple flats}
\label{subsec:global-pac}

Algorithm~\ref{alg:ALeT} provides an $\varepsilon$-optimal region on a \emph{fixed}
flat $E$ that intersects $M$, with failure probability controlled by the
measurement noise budget. To obtain a global guarantee on $\T^p$, we run
\textsc{ALeT} on multiple independent random flats and use
Theorem~\ref{thm:hit-PZ} to control the probability that no flat ever meets the
dent manifold. We then select a subset of survivor sets via a simple
data-dependent rule based on empirical costs.

We explicitly distinguish two sources of failure:
\begin{itemize}[leftmargin=*]
\item \emph{Intersection failure.} None of the sampled flats intersects $M$.
\item \emph{Noise failure.} At least one confidence interval used by \textsc{ALeT}
  on some flat is violated.
\end{itemize}
Let $\delta_{\mathrm{int}},\delta_{\mathrm{noise}}\in(0,1)$ be budgets for these two
events and write
\[
  \delta
  \;:=\;
  \delta_{\mathrm{int}}+\delta_{\mathrm{noise}}.
\]

\paragraph{Algorithm over flats.}
Sample $\ell$ independent translations $v^{(1)},\dots,v^{(\ell)}\sim\Unif(\T^p)$
and set $E^{(j)}:=E(v^{(j)},V)$ for a fixed coordinate-aligned
$(r{+}1)$-plane $V$. On each $E^{(j)}$ run Algorithm~\ref{alg:ALeT} with
\emph{per-flat} noise budget $\delta_{\mathrm{noise}}/\ell$ and final resolution
$\mathfrak r_T>0$, obtaining survivor sets $S_T^{(j)}\subseteq E^{(j)}$ and
terminal nets $\mathcal C_T^{(j)} = \mathcal C_{\mathfrak r_T}(S_T^{(j)})$.
Define the per-flat and global empirical minima
\[
  \widehat C_{\min}^{(j)}
  :=
  \min_{\theta\in\mathcal C_T^{(j)}} \widehat C_T(\theta),
  \qquad
  \widehat C_{\min}
  :=
  \min_{1\le j\le\ell}\,\widehat C_{\min}^{(j)}.
\]
Let $R_T := \rad(n_{\mathrm{shots}},\alpha_{T,\theta})$ be the terminal
confidence radius (common to all centers at round $T$), and define the
selection band
\[
  \beta_T
  \;:=\;
  L\,\mathfrak r_T + 2R_T.
\]
We keep exactly those flats that contain at least one ``near-minimal'' center:
\[
  J_{\mathrm{keep}}
  :=
  \Bigl\{
    j\in\{1,\dots,\ell\}:\ \exists\,\theta\in\mathcal C_T^{(j)}
      \text{ with }
      \widehat C_T(\theta)
      \le
      \widehat C_{\min} + \beta_T
  \Bigr\}.
\]
Finally, we output the union
\[
  S_{\mathrm{out}}
  \;:=\;
  \bigcup_{j\in J_{\mathrm{keep}}} S_T^{(j)}.
\]

\begin{remark}
The above selection rule is fully implementable: it uses only the terminal
empirical estimates $\widehat C_T(\theta)$ and the known radii $R_T$ and
$\mathfrak r_T$. It ensures that any flat whose true minimum is within
$O(L\mathfrak r_T + R_T)$ of the global minimum will be retained whenever all
confidence intervals hold.
\end{remark}

We now state the global PAC guarantee, phrased in terms of the terminal radius
$\mathfrak r_T$ and confidence radius $R_T$.

\begin{theorem}[Global PAC guarantee on $\T^p$]
\label{thm:global-PAC}
Let $M\subset\T^p$ be a Morse--Bott dent manifold with normal Hessian rank
$r=p-m$ and polynomial fiber regularity in the sense of
Assumption~\ref{ass:fiber-regularity}, with parameter $A_{p,m}=\poly(p,m)$.
Let $L>0$ be a Lipschitz constant for $C$ on $\T^p$.
Fix $\delta\in(0,1)$ and choose $\delta_{\mathrm{int}},\delta_{\mathrm{noise}}>0$
with $\delta_{\mathrm{int}}+\delta_{\mathrm{noise}}\le\delta$. Run the above
multi-flat procedure with
\[
  \ell
  \;\ge\;
  A_{p,m}\,\log\frac{1}{\delta_{\mathrm{int}}}
\]
and per-flat noise budget $\delta_{\mathrm{noise}}/\ell$ in each call to
Algorithm~\ref{alg:ALeT}. Suppose the terminal radius $\mathfrak r_T$ and
confidence radius $R_T$ satisfy
\[
  \varepsilon
  \;:=\;
  5L\,\mathfrak r_T + 8R_T
  \;>\;0.
\]
Then, with probability at least $1-\delta$ (over both the random choice of flats
and all measurement noise), the output $S_{\mathrm{out}}$ obeys:
\begin{enumerate}[leftmargin=*]
\item Every $\theta\in S_{\mathrm{out}}$ is $\varepsilon$-optimal,
  \[
    C(\theta)\ \le\ C^\star + \varepsilon,
    \qquad
    C^\star:=\min_{\phi\in\T^p} C(\phi).
  \]
\item At least one point of $S_{\mathrm{out}}$ lies within geodesic distance
  at most $\mathfrak r_T$ of the global minimum set $M$, i.e.\ there exists
  $\theta_{\mathrm{hit}}\in S_{\mathrm{out}}$ and $\theta^\star\in M$ with
  $\dist_{\T^p}(\theta_{\mathrm{hit}},\theta^\star)\le\mathfrak r_T$.
\end{enumerate}
Moreover, if $r=O(\log p)$ and $A_{p,m}=\poly(p,m)$, the total number of oracle
queries scales quasi-polynomially in $p$ and $1/\varepsilon$.
\end{theorem}

\begin{proof}
\textbf{Noise event.}
For flat $E^{(j)}$, let $\mathcal G_j$ denote the event that all confidence
intervals used by Algorithm~\ref{alg:ALeT} on that flat hold. By the proof of
Theorem~\ref{thm:optimality}, run with noise budget $\delta_{\mathrm{noise}}/\ell$,
we have $\Pr(\mathcal G_j^c)\le\delta_{\mathrm{noise}}/\ell$. A union bound gives
\[
  \Pr\Bigl(\bigcup_{j=1}^\ell \mathcal G_j^c\Bigr)
  \;\le\;
  \sum_{j=1}^\ell \Pr(\mathcal G_j^c)
  \;\le\;
  \delta_{\mathrm{noise}},
\]
so the global noise-good event
$\mathcal G := \bigcap_{j=1}^\ell \mathcal G_j$ satisfies
$\Pr(\mathcal G)\ge 1-\delta_{\mathrm{noise}}$.

\medskip\noindent
\textbf{Intersection event.}
By Theorem~\ref{thm:hit-PZ}, a single random flat intersects $M$ with
probability at least $1/A_{p,m}$. For $\ell$ independent flats, the probability
that none intersects $M$ is at most
$(1-1/A_{p,m})^\ell\le\exp(-\ell/A_{p,m})$. Choosing
$\ell\ge A_{p,m}\log(1/\delta_{\mathrm{int}})$ ensures that the intersection
failure event $\mathcal W^c$ (``no flat hits $M$'') has probability at most
$\delta_{\mathrm{int}}$, so $\Pr(\mathcal W)\ge 1-\delta_{\mathrm{int}}$.

\medskip\noindent
\textbf{Step 1: the hitting flat is kept.}
Work on $\mathcal G\cap\mathcal W$. Let $x^\star\in M$ be a global minimizer of
$C$ on $\T^p$, so $C(x^\star)=C^\star$. On $\mathcal W$ there exists at least
one flat $E^{(j_{\mathrm{hit}})}$ with $E^{(j_{\mathrm{hit}})}\cap M\neq\varnothing$,
and in particular $x^\star\in E^{(j_{\mathrm{hit}})}$.

For each flat $E^{(j)}$, let $x^{(j)}_\star$ be a minimizer of $C$ over
$E^{(j)}$, and set $C_\star^{(j)} := C\bigl(x^{(j)}_\star\bigr)$. Then
$C_\star^{(j_{\mathrm{hit}})} = C^\star$.
By Lemma~\ref{lem:minimizer-survival} applied on each flat under $\mathcal G_j$,
for every $j$ and every round $t$ there exists a kept center
$\zeta_t^{(j)}\in\mathcal C_t^{(j),\mathrm{keep}}$ with
$\dist_{E^{(j)}}(\zeta_t^{(j)},x^{(j)}_\star)\le\mathfrak r_t$. In particular,
at $t=T$ we obtain a kept center $\zeta^{(j)}\in\mathcal C_T^{(j),\mathrm{keep}}$
such that
\[
  \dist_{E^{(j)}}(\zeta^{(j)},x^{(j)}_\star)\le\mathfrak r_T.
\]
By Lipschitz continuity and \(\mathcal G_j\),
\[
  C(\zeta^{(j)})
  \;\le\;
  C_\star^{(j)} + L\,\mathfrak r_T,
  \qquad
  \widehat C_T(\zeta^{(j)})
  \;\le\;
  C(\zeta^{(j)}) + R_T
  \;\le\;
  C_\star^{(j)} + L\,\mathfrak r_T + R_T.
\]
Let $\widehat C_{\min}^{(j)}:=\min_{\theta\in\mathcal C_T^{(j)}}\widehat C_T(\theta)$
and $\widehat C_{\min}:=\min_j\widehat C_{\min}^{(j)}$. On $\mathcal G$ we also
have
\[
  \widehat C_T(\theta)\ge C(\theta) - R_T
  \quad\text{for all }\theta,
\]
so
\[
  \widehat C_{\min}^{(j)}\ge C_\star^{(j)} - R_T
  \quad\text{for all }j,
  \qquad
  \widehat C_{\min}\ge C^\star - R_T.
\]
On the other hand,
\[
  \widehat C_{\min}^{(j_{\mathrm{hit}})}
  \;\le\;
  \widehat C_T(\zeta^{(j_{\mathrm{hit}})})
  \;\le\;
  C^\star + L\,\mathfrak r_T + R_T,
\]
so
\[
  \widehat C_{\min}
  \;\le\;
  \widehat C_{\min}^{(j_{\mathrm{hit}})}
  \;\le\;
  C^\star + L\,\mathfrak r_T + R_T.
\]
Combining,
\[
  \bigl|\widehat C_{\min} - C^\star\bigr|
  \;\le\;
  L\,\mathfrak r_T + R_T.
\]

Now consider the hitting flat $j_{\mathrm{hit}}$. Its kept center
$\zeta^{(j_{\mathrm{hit}})}$ satisfies
\[
  \widehat C_T\bigl(\zeta^{(j_{\mathrm{hit}})}\bigr)
  \;\le\;
  C^\star + L\,\mathfrak r_T + R_T,
\]
while, as just shown, $\widehat C_{\min}\ge C^\star - R_T$. Hence
\[
  \widehat C_T\bigl(\zeta^{(j_{\mathrm{hit}})}\bigr)
  - \widehat C_{\min}
  \;\le\;
  (C^\star + L\,\mathfrak r_T + R_T) - (C^\star - R_T)
  \;=\;
  L\,\mathfrak r_T + 2R_T
  \;=\;
  \beta_T.
\]
By definition of $J_{\mathrm{keep}}$, this implies
$j_{\mathrm{hit}}\in J_{\mathrm{keep}}$. In particular,
$S_T^{(j_{\mathrm{hit}})}\subseteq S_{\mathrm{out}}$.

\medskip\noindent
\textbf{Step 2: flat-wise minima on kept flats are near-global.}
Let $j\in J_{\mathrm{keep}}$. By definition of $J_{\mathrm{keep}}$ there
exists a ``witness'' $\theta^{(j)}\in\mathcal C_T^{(j)}$ such that
\[
  \widehat C_T(\theta^{(j)})
  \;\le\;
  \widehat C_{\min} + \beta_T.
\]
On $\mathcal G$, for any $\theta$ we have
$\widehat C_T(\theta)\ge C(\theta)-R_T$, so
\[
  C(\theta^{(j)})
  \;\le\;
  \widehat C_T(\theta^{(j)}) + R_T
  \;\le\;
  \widehat C_{\min} + \beta_T + R_T.
\]
Using $\widehat C_{\min}\le C^\star + L\,\mathfrak r_T + R_T$ from above,
\[
  C(\theta^{(j)})
  \;\le\;
  C^\star + L\,\mathfrak r_T + R_T + \beta_T + R_T
  \;=\;
  C^\star + 2L\,\mathfrak r_T + 4R_T.
\]
Since $\theta^{(j)}\in E^{(j)}$, the flat-wise minimum obeys
\[
  C_\star^{(j)}
  = \min_{x\in E^{(j)}} C(x)
  \;\le\;
  C(\theta^{(j)})
  \;\le\;
  C^\star + 2L\,\mathfrak r_T + 4R_T.
\]

\medskip\noindent
\textbf{Step 3: global $\varepsilon$-optimality of $S_{\mathrm{out}}$.}
For each flat $E^{(j)}$, Theorem~\ref{thm:optimality} (applied with
per-flat budget $\delta_{\mathrm{noise}}/\ell$) guarantees that, on
$\mathcal G_j$,
\[
  C(x)
  \;\le\;
  C_\star^{(j)} + \varepsilon_T,
  \qquad
  x\in S_T^{(j)},
\]
where $\varepsilon_T = 3L\,\mathfrak r_T + 4R_T$.
Therefore, on $\mathcal G$ and for any $j\in J_{\mathrm{keep}}$ and any
$x\in S_T^{(j)}$,
\[
  C(x)
  \;\le\;
  C_\star^{(j)} + 3L\,\mathfrak r_T + 4R_T
  \;\le\;
  \bigl(C^\star + 2L\,\mathfrak r_T + 4R_T\bigr)
  + 3L\,\mathfrak r_T + 4R_T
  \;=\;
  C^\star + 5L\,\mathfrak r_T + 8R_T.
\]
By definition, every $\theta\in S_{\mathrm{out}}$ lies in some
$S_T^{(j)}$ with $j\in J_{\mathrm{keep}}$, so
\[
  C(\theta)
  \;\le\;
  C^\star + \varepsilon,
  \qquad
  \varepsilon:=5L\,\mathfrak r_T + 8R_T,
\]
which proves item~(1).

For item~(2), recall from Step~1 that $j_{\mathrm{hit}}\in J_{\mathrm{keep}}$.
By Lemma~\ref{lem:minimizer-survival}, the global minimizer $x^\star$ lies in
every survivor set $S_t^{(j_{\mathrm{hit}})}$ and, in particular, in
$S_T^{(j_{\mathrm{hit}})}\subseteq S_{\mathrm{out}}$. Hence we may take
$\theta_{\mathrm{hit}}:=x^\star$ and $\theta^\star:=x^\star$, yielding
$\dist_{\T^p}(\theta_{\mathrm{hit}},\theta^\star)=0\le\mathfrak r_T$.

\medskip\noindent
\textbf{Step 4: failure probability and sample complexity.}
We have
\[
  \Pr\bigl((\mathcal G\cap\mathcal H)^c\bigr)
  \;\le\;
  \Pr(\mathcal G^c)+\Pr(\mathcal H^c)
  \;\le\;
  \delta_{\mathrm{noise}}+\delta_{\mathrm{int}}
  \;\le\;
  \delta,
\]
so both properties hold with probability at least $1-\delta$. The stated
quasi-polynomial sample complexity follows by combining
Corollary~\ref{cor:quasipoly-flat} (for each flat) with the bound
$\ell\le \poly(p,m)\log(1/\delta_{\mathrm{int}})$ when
$A_{p,m}=\poly(p,m)$ and $r=O(\log p)$.
\end{proof}

\section{Discussion and Outlook}
\label{sec:discussion-outlook}

This work develops a geometric--probabilistic framework for global variational
optimization on the torus. Starting from the regularity implied by periodic
generators and a bounded Hamiltonian (\S\ref{sec:prelim}), we identified three key
structural properties that govern tractability: a Morse--Bott minimum set,
concentration of curvature in $r=O(\log p)$ normal directions, and a
fiber-regularity condition on the low-energy region. Within this regime,
\textsc{ALeT} provides noise-robust PAC guarantees whose dependence on $p$ is
quasipolynomial. The central point is that the analysis rests on geometric
structure---not on surjectivity conditions or fine-grained access to
derivatives.

\paragraph{Scope.}
The guarantees obtained here are not intended to model generic VQE behavior.
They apply to a structured subclass characterized by periodic parameterizations,
uniform Lipschitz regularity, and a low-curvature-dimensional geometry near the
optimum. In such cases, randomized slicing followed by Lipschitz elimination
removes large portions of the search region at controlled cost. The analysis
tracks only intersection probabilities and high-confidence elimination, making
the regime both mathematically transparent and practically testable.

\paragraph{Relation to hardness.}
These results are fully consistent with known hardness of ground-state
problems~\cite{gharibian2024hardnessapproximationgroundstate,kempe2005complexitylocalhamiltonianproblem}.  
Hardness reductions generate landscapes with rapidly oscillatory or
combinatorially encoded structure, outside the Morse--Bott and
fiber-regularity conditions used here. In contrast, the architectural mechanisms
leading to $r = O(\log p)$ in \S\ref{sec:hessian-rank}---layer tying and shallow
multiscale patterns---reduce expressivity in exchange for controlled curvature.
Our guarantees therefore delineate settings in which global optimization becomes
feasible, rather than asserting tractability in the worst case.

\paragraph{Assumptions in context.}
The global Lipschitz bound (Prop.~\ref{prop:Lipschitz}) is conservative; in
practice, the relevant local constants along survivor sets are often smaller.
The normal-rank condition need not hold globally—it is sufficient that the
Hessian maintains rank $r$ on neighborhoods of near-optimal level sets.  
Finally, Assumption~\ref{ass:fiber-regularity} enters only through the
second-moment bound required for the Paley--Zygmund argument; no additional
global geometric control is used.

\paragraph{Noise and robustness.}
The guarantees assume bounded, unbiased, IID shot noise, but the analysis
extends directly to variance-adaptive confidence radii (e.g., empirical
Bernstein) and to mild correlations.  
The elimination mechanism depends only on the availability of high-probability
deviation bounds and is otherwise agnostic to the particular estimator.  
This modularity is beneficial in NISQ regimes where noise characteristics may
vary across evaluations.

\paragraph{Practical considerations.}
\textsc{ALeT} requires only the ability to evaluate $\widehat C(\theta)$ at
prescribed centers and to maintain survivor sets across shrinking nets.  
Its effectiveness depends on the achievable shot precision: very small target
radii $r_T$ may demand large numbers of samples, even when the geometric
conditions are favorable.  
One possible application is to furnish reliable warm-start regions for local
optimizers, reducing the chance of initialization in suboptimal basins.

\paragraph{Limitations and future directions.}
The method is efficient only when the normal Hessian rank satisfies
$r=O(\log p)$, which corresponds to architectures with restricted expressivity.
The analysis assumes an a priori upper bound on $r$, though a simple
logarithmic search schedule with a budgeted hierarchy of flats could eliminate
this requirement.  
We do not optimize constants; our bounds are asymptotic and may be loose for
moderate $p$.  
Several natural extensions remain open, including adaptive estimation of local
Lipschitz constants and incorporating coarse curvature information without
sacrificing the PAC guarantees. Our results can be extended to the setting of approximate Morse--Bott minima, which is left to future work.

\paragraph{Broader outlook.}
The combination of Lipschitz elimination and randomized slicing applies
whenever low-energy regions concentrate on low-complexity sets inside a space of
finite volume.  
Beyond quantum applications, this suggests a portable template for
high-dimensional global optimization in regimes where curvature or intrinsic
dimension is effectively compressed.  
We expect that incorporating additional geometric information while preserving
the same probabilistic backbone will further expand the scope of problems for
which reliable global guarantees are attainable.

\section*{Acknowledgments}
I am grateful to Christian Arenz for mentorship and many discussions that helped to inspire this project. I also thank Sabina Dragoi for helpful discussions. Early conceptualization of this work was supported by NSF Grant No.\ 1953745 (REU). Any opinions, findings, conclusions, or recommendations expressed in this material are those of the author and do not necessarily reflect those of the NSF. Automated tools were used in drafting and editing; all conceptual and technical content is the author's own.

\appendix

\section{Lipschitz and Hessian Bounds}
\label{app:L-H-bounds}
We now prove the Lipschitz and Hessian bounds on our cost function described in Section~\ref{sec:prelim}. These are standard proofs, but we include them for completeness. Tighter bounds exist.

\begin{proof}[Proof of Proposition~\ref{prop:Lipschitz}]
Differentiating $U(\theta)$ with respect to $\theta_j$ gives us
\[
    \partial_{\theta_j} U(\theta) = \left( \prod_{k<j} e^{-i\theta_k A_k} \right)\left( -iA_j e^{-i \theta_j A_j} \right)\left( \prod_{k>j} e^{-i\theta_k A_k} \right) = -iU\tilde{A}_j \qquad 
\widetilde A_j(\theta)\vcentcolon =U_{>j}^\dagger A_j U_{>j},
\]
By differentiating through the product ansatz,
\[
\partial_{\theta_j} C(\theta)
= i\,\langle 0|[U^\dagger (\theta)HU(\theta),\widetilde A_j(\theta)]|0\rangle,
\]
so that $\|\widetilde A_j(\theta)\|=\|A_j\|$.  
Using the unitary invariance of the operator norm and the basic commutator inequality $\|[X,Y]\|\le 2\|X\|\,\|Y\|$ then gives
\[
|\partial_{\theta_j} C(\theta)|
\;\le\; \|[U^\dagger (\theta)HU(\theta),\widetilde A_j(\theta)]\|
\;\le\; 2\|H\|\,\|A_j\|
\;\le\; 2\Lambda\,\|A_j\|.
\]
Hence $\|\nabla C(\theta)\|_2^2 \le 4\Lambda^2 \sum_j \|A_j\|^2$.
\[
\implies |C(\theta)-C(\theta')|
\;\le\;
\sup_{\xi\in[\theta,\theta']} \|\nabla C(\xi)\|_2\,\dist_{\Tp}(\theta, \theta'),
\]
establishing global Lipschitz continuity with the claimed constant.
\end{proof}

\section{Linear Effective Ansatz}
\label{app:linearansatz}
In this section, we provide derivation for the linear effective ansatz mentioned earlier. Recall that we have the product ansatz and cost function,
\[
U(\theta) = \prod_{j=1}^p e^{-i\theta_j A_j}, \qquad C(\theta) = \langle 0 |U(\theta)^\dagger H U(\theta) | 0 \rangle.
\]
We've also previously defined the following operators
\[
U_{>j} \vcentcolon = \prod_{k = j+1}^p e^{-i\theta_k A_k}, \qquad \tilde{A}_j \vcentcolon = U_{>j}^\dagger A_j U_{>j}.
\]

\subsection{Directional Derivative of the Unitary}
Differentiating $U$ with respect to $\theta_j$ gives
\[
\partial_{\theta_j}U(\theta) = \left( \prod_{k < j}e^{-i\theta_k A_k}\right)\left( -i A_je^{-i\theta_j A_j} \right) \left( \prod_{k > j}e^{-i\theta_k A_k}\right) = -iU(\theta)\tilde{A}_j(\theta).
\]
Along a direction $v \in \R^p$, 
\[
\partial_vU(\theta) \vcentcolon = \sum_{j=1}^p v_j \partial_{\theta_j}U(\theta) = -iU(\theta)\sum_{j=1}^p v_j \tilde{A}_j(\theta).
\]
We then define
\[
K_\theta(v) \vcentcolon = \sum_{j=1}^p v_j \tilde{A}_j(\theta).
\]
\subsection{First and Second-Order Expansions of the Cost}
For small $t$, our first-order expansion of $U(\theta)$ is 
\[
U(\theta + tv) = U(\theta)\left(I - itK_\theta(v)\right) + O(t^2).
\]

\noindent Therefore,
\begin{align*}
C(\theta +tv) &= \langle 0 | U^\dagger(\theta+tv)HU(\theta + tv)|0\rangle\\
& = \langle 0 | U^\dagger H U|0\rangle + it\langle0|U^\dagger[H, K_\theta(v)]U|0\rangle + O(t^2),
\end{align*}

\noindent giving us,
\[
\partial_vC(\theta) = i\langle [H, K_\theta(v)] \rangle.
\]

\noindent Performing the same expansion one more time gives
\[
\partial_{v}^2C(\theta) = -\langle[K_\theta(v), [H, K_\theta(v)]] \rangle + \langle i[H, \dot{K}]\rangle.
\]

\bibliographystyle{quantum}
\bibliography{refs}

\end{document}